\documentclass{article}

\usepackage{color}
\usepackage{graphicx}
\usepackage{amsmath}
\usepackage{amsfonts}
\usepackage{amsthm}
\usepackage{amssymb}
\usepackage{algorithm}
\usepackage[noend]{algorithmic}
\usepackage{textcomp}

\newlength{\figwidth}
\newtheorem{definition}{Definition}
\newtheorem{theorem}{Theorem}
\newtheorem{lemma}{Lemma}
\newtheorem{corollary}{Corollary}

\newcommand{\nop}[1]{}
\newcommand{\fds}{\textsc{fd}s}
\newcommand{\fd}{\textsc{fd}}
\newcommand{\cfds}{\textsc{cfd}s}

\newcommand{\md}{\textsc{md}}
\newcommand{\mds}{\textsc{md}s}
\newcommand{\ea}{\textsc{ea}}
\newcommand{\eps}{\textsc{eps}}
\newcommand{\epsc}{\textsc{epsc}}
\newcommand{\aps}{\textsc{aps}}
\newcommand{\apsi}{\textsc{apsi}}

\begin{document}

\title{Discovering Matching Dependencies}


\maketitle

\begin{abstract}

The concept of \emph{matching dependencies} (\mds) is recently proposed for
specifying matching rules for object identification. Similar to the
functional dependencies (with conditions), \mds\ can also be applied
to various data quality applications such as violation detection.
In this paper, we study the problem of discovering matching dependencies from a given database instance. First, we formally define the measures, support and confidence, for
evaluating utility of \mds\ in the given database instance. Then, we study the
discovery of \mds\ with certain utility requirements of support and confidence.
Exact algorithms are developed, together with pruning strategies to improve the time performance.
Since the exact algorithm has to traverse all the data during the computation, we propose an approximate solution which only use some of the data. A bound of relative errors introduced by the approximation is also developed.
Finally, our experimental evaluation demonstrates the efficiency of
the proposed methods.

\end{abstract}




\section{Introduction}\label{sect_introduction}

Recently, data quality has become a hot topic in database community due to huge amount of ``dirty'' data originated from different resources (see~\cite{DBLP:books/sp/dcsa/Batini06} for a survey). These data are often ``dirty'', including inconsistencies, conflicts, and errors, due to various erroneous introduced by human and machines. In addition to cost of dealing the huge volume of data, manually detecting and removing ``dirty'' data is definitely out of practice because human proposed cleaning methods may introduce inconsistencies again. Therefore, data dependencies, which have been widely used in the relational database design to set up the integrity constraints, have been revisited and revised to capture wider inconsistencies in the data.
For example, consider a $\mathsf{Contacts}$ relation with the schema:
$$\mathsf{Contacts(SIN, Name, CC, ZIP, City, Street)}$$
The following functional dependency $\mathsf{fd}$ specifies a constraint that for any two tuples in $\mathsf{Contacts}$, if they have the same $\mathsf{ZIP}$ code, then these two tuples have the same $\mathsf{City}$ as well.
Recently, \emph{functional dependencies} (\fds) have been extended to \emph{conditional functional dependencies} (\cfds)~\cite{DBLP:conf/icde/BohannonFGJK07}, i.e., \fds\ with conditions, which have more expressive power. The basic idea of these extensions is making the \fds, originally hold for the whole table, valid only for a set of tuples. For example, the following $\mathsf{cfd}$ specifies that only in the condition of country code $\mathsf{CC=44}$, if two tuples have the same $\mathsf{ZIP}$, then they must have same $\mathsf{Street}$ as well.
\begin{eqnarray*}
\mathsf{fd}  & : & [\mathsf{ZIP}]\rightarrow[\mathsf{City}]\\
\mathsf{cfd} & : & [\mathsf{ZIP}, \mathsf{CC=44}]\rightarrow[\mathsf{Street}]
\end{eqnarray*}
These dependency constraints can be used to detect data violations~\cite{DBLP:conf/vldb/CongFGJM07}.
For instance, we can use the above $\mathsf{fd}$ to detect violations in an instance of $\mathsf{Contacts}$ in Table~\ref{table_data_relation}. For the tuples $t_5$ and $t_6$ with the same values of
$\mathsf{ZIP}=\mathsf{021}$, they have different values of
$\mathsf{City}$, which are then detected as violations of the above $\mathsf{fd}$.

Although functional dependencies (and their extension with conditions) are very useful in determining data inconsistency and repairing the ``dirty'' data~\cite{DBLP:conf/vldb/CongFGJM07}, they check the specified attribute value agreement based on \emph{exact match}. For example, with the above $\mathsf{cfd}$, tuples that have $\mathsf{CC=44}$ and the same value on $\mathsf{ZIP}$ attribute will be checked to see whether they have exactly matched values on $\mathsf{Street}$. Obviously, this strict exact match constraint limits usage of \fds\ and \cfds, since real-world information often have various representation formats. For example, the tuples $t_2$ and $t_3$ in $\mathsf{Contacts}$ table will be detected as ``violations'' of the $\mathsf{cfd}$, since they have ``different'' $\mathsf{Street}$ values but agree on $\mathsf{ZIP}$ and $\mathsf{CC=44}$. However, ``No.2, Central Rd.'' and ``\#2, Central Rd.'' are exactly the ``same'' street in the real-world with different representation formats.

To make dependencies adapt to this real-world scenario, i.e., to be tolerant of various representation formats,  Fan~\cite{DBLP:conf/pods/Fan08} proposed a new concept of data dependencies, called \emph{matching dependencies} (\mds). Informally, a matching dependency targets on the fuzzy values like
text attributes and defines the dependency between two set of attributes according to their matching quality measured by some matching operators
(see~\cite{DBLP:journals/expert/BilenkoMCRF03} for a survey), such as \emph{Euclidean distance} and \emph{cosine similarity}. Again, in $\mathsf{Contacts}$ example, we may have a \md\ as
\begin{eqnarray*}
\mathsf{md_1} & : & ([\mathsf{Street}]\rightarrow[\mathsf{City}], <0.8, 0.7>)
\end{eqnarray*}
which states that for any two tuples from $\mathsf{Contacts}$, if they agree on attribute $\mathsf{Street}$ (the matching similarity, e.g. \emph{cosine similarity}, on the attribute $\mathsf{Street}$ is greater than a threshold $0.8$), then the corresponding $\mathsf{City}$ attribute should match as well (i.e. similarity on $\mathsf{City}$ is greater than the corresponding threshold $0.7$).

 \begin{table}[t]
  \caption{Example of $\mathsf{Contacts}$ relation $\mathcal{R}$}
 \label{table_data_relation}
 \centering\small
 \begin{tabular}{|c|c|c|c|c|c|c}
 \cline{1-6}
 SIN & Name & CC & ZIP & City & Street & \\
 \cline{1-6} \cline{1-6}
 584 & Claire Green              & 44 & 606     & Chicago   & No.2, Central Rd. & $t_1$\\ \cline{1-6}
 584 & Claire Gree\underline{m}  & 44 & 606     & Chicago   & No.2, Central Rd. & $t_2$\\ \cline{1-6}
 584 & Claire Gree\underline{~~} & 44 & 606     & Chicago   & \#2,  Central Rd. & $t_3$\\ \cline{1-6}
 265 & Jason Smith               & 01 & 021     & Boston    & No.3, Central Rd. & $t_4$\\ \cline{1-6}
 265 & J. Smith                  & 01 & 021     & Boston    & \#3,  Central Rd. & $t_5$\\ \cline{1-6}
 939 & W. J. Smith               & 01 & 021     & Chicago   & \#3,  Central Rd. & $t_6$\\ \cline{1-6}
 \end{tabular}
 \end{table}

Similar to the \fds\ related techniques, \mds\ can be applied in
many tasks as well~\cite{DBLP:conf/pods/Fan08}. For example, in data cleaning, we can
also use \mds\ to detect the inconsistent data, that is, data do not
follow the constraint (rule) specified by \mds.
For example, according to the above $\mathsf{md}$ example, for any two tuples $t_i$ and $t_j$ having similarity greater than $0.8$ on $\mathsf{Street}$, they should be matched on $\mathsf{City}$ as well (similarity $\geq 0.7$). If their $\mathsf{City}$ similarity is less than $0.7$, then there must be something wrong in $t_i$ and $t_j$, i.e., inconsistency. Such inconsistency on text attributes cannot be detected by using \fds\ and extensions based on exact matching.
In addition to locating the inconsistent data, object identification, another important
work for data cleaning,  can also employ \mds\ as matching rules~\cite{DBLP:journals/pvldb/FanLJM09}.
For instance, according to
$$\mathsf{md_2} : ([\mathsf{Name}, \mathsf{Street}] \rightarrow [\mathsf{SIN}],<0.9, 0.9, 1.0>)$$
if two tuples have high similarities on $\mathsf{Name}$ and $\mathsf{Street}$ (both similarities are greater than 0.9), then
these two tuples probably denote the same person in the real world, i.e., having the same $\mathsf{SIN}$.

Though the concept of matching dependencies is given
in~\cite{DBLP:conf/pods/Fan08}, the authors did not discuss how to discover
useful \mds. In fact, given a database instance, there
are enormous \mds\ that can be discovered if we set different similarity thresholds on attributes. Note that if all thresholds are set to $1.0$, \mds\ have the same semantics as traditional \fds, in other words, traditional \fds\ are special cases of \mds. For instance, the above $\mathsf{fd}$ can be represented by a \md\ $([\mathsf{ZIP}]\rightarrow[\mathsf{City}], <1.0, 1.0>)$. Clearly, not all the settings of thresholds for \mds\ are useful.

The utility of \mds\ in the above applications is often evaluated by \emph{confidence} and \emph{support}. Specifically, we consider a \md\ of a relation $\mathcal{R}$, denoted by $\varphi(X \rightarrow Y, \lambda)$, where $X$ and $Y$ are the attribute sets of $\mathcal{R}$, $\lambda$ is a pattern specifying different similarity thresholds on each attribute in $X$ and $Y$. Let $\lambda_X$ and $\lambda_Y$ be the projections of thresholds in pattern $\lambda$ on the attributes $X$ and $Y$ respectively.
The \emph{support} of $\varphi$ is the proportion of tuple pairs whose matching similarities are higher than the thresholds in $\varphi$ on both attributes of $X$ and $Y$. The \emph{confidence} is the ratio of tuple pairs whose matching similarities satisfy $\lambda_X$ also satisfying $\lambda_Y$.
In real applications like inconsistency detection, in order to achieve high detection accuracy, we would like to use \mds\ with high confidence. On the other hand, if users need high recall of detection, then \mds\ with high support are preferred.
Intuitively, we would like to discover those \mds\ with high support, high confidence and high matching quality.
Therefore, in this work, we would like to discover proper settings of matching similarity thresholds for \mds, which can satisfy users' utility requirements of support and confidence.

\paragraph{Contributions}

In this paper, given a relation instance and $X\rightarrow Y$, we study the issues of
discovering matching dependencies on the given $X\rightarrow Y$. Our main contributions are
summarized as follows:


First, we propose the utility evaluation of
\emph{matching dependencies}. Specifically, the
confidence and support evaluations of \mds\ are formally defined.
To the best of our knowledge, this is the first paper to study the utility evaluation and discovery of \mds.

Second, we study the exact algorithms for discovering \mds. The \mds\ discovery problem is to find settings of matching similarity thresholds on
attributes $X$ and $Y$ for \mds\ that can satisfy the required confidence and
support. We first present an exact solution and then study pruning
strategies by the minimum requirements of support and confidence.

Third, we study the approximation algorithms for discovering \mds. Since the exact algorithm has to traverse all the data during the computation, we propose an approximate solution which only use some of the data. A bound of relative errors introduced by the approximation is developed. Moreover, we also develop a strategy of early termination in individual step.

Finally, we report an extensive experimental evaluation. The
proposed algorithms on discovering \mds\ are studied. Our
pruning strategies can significantly improve the efficiency in
discovering \mds.

The remainder of this paper is organized as follows. First, we introduce some related work in Section~\ref{sect_related}. Then, Section~\ref{sect_model} presents the utility measures for \mds, including support and confidence. In Section~\ref{sect_optimal_md}, we develop the exact algorithm for discovering \mds\ and study the corresponding pruning strategies. In Section~\ref{sect_approximation}, we present the approximation algorithm with bounded relative errors. In Section~\ref{sect_experiment}, we report our extensive experimental evaluation. Finally, Section~\ref{sect_conclusion} concludes this paper.
Table~\ref{table_notations} lists the frequently used notations in this paper.

 \begin{table}[t]
 \caption{Notations}
 \label{table_notations}
 \centering
 \begin{tabular}{cl}
 \hline\noalign{\smallskip} Symbol & Description \\
 \hline \noalign{\smallskip}
 $\varphi$ & Matching dependency, \md\ \\ \noalign{\smallskip}
 $\lambda$ & Threshold pattern, of matching similarity \\ \noalign{\smallskip}
 $\mathcal{C}_t$ & Candidate set, of total $c$ threshold patterns\\ \noalign{\smallskip}
 $\eta_s$ & Minimum requirement, of support  \\ \noalign{\smallskip}
 $\eta_c$ & Minimum requirement, of confidence \\ \noalign{\smallskip}
 $\mathcal{R}$ & Original relation, of $N$ data tuples $t$ \\ \noalign{\smallskip}
 $\mathcal{D}$ & Statistical distribution, of $n$ statistical tuples $s$ \\ \noalign{\smallskip}
 \hline
 \end{tabular}
 \end{table}

\section{Related Work}\label{sect_related}

Traditional dependencies, such as functional dependencies
(\fds) and inclusion dependencies (\textsc{ind}s) for the schema
design~\cite{DBLP:books/aw/AbiteboulHV95}, are revisited for new
applications like improving the quality of data. The conditional
functional dependencies (\cfds) are first proposed
in~\cite{DBLP:conf/icde/BohannonFGJK07} for data cleaning. Cong et
al.~\cite{DBLP:conf/vldb/CongFGJM07} study the detecting and
repairing methods of violation by \cfds. Fan et
al.~\cite{DBLP:journals/pvldb/FanMHLW08} investigate the propagation
of \cfds\ for data integration. Bravo et
al.~\cite{DBLP:conf/icde/BravoFGM08} propose an extension of \cfds\
by employing disjunction and negation. Golab et
al.~\cite{DBLP:journals/pvldb/GolabKKSY08} define a range tableau
for \cfds, where each value is a range similar to the concept of
matching similarity intervals in our study. In addition, Bravo et al.~\cite{DBLP:conf/vldb/BravoFM07}  propose conditional inclusion dependency (\textsc{cind}s), which are useful not only in data cleaning, but are also in contextual schema matching.
Ilyas et al.\cite{DBLP:conf/sigmod/IlyasMHBA04} study a novel soft \fd, which is also a generalization of the
classical notion of a hard \fd\ where the value of $X$ completely
determines the value of $Y$. In a soft \fd, the value of $X$ determines the value of $Y$
not with certainty, but merely with high probability.


The confidence and support measures are widely used in discovering approximate functional dependencies~\cite{DBLP:journals/cj/HuhtalaKPT99,DBLP:journals/jamds/KingL03} and evaluating
\cfds~\cite{DBLP:journals/pvldb/GolabKKSY08,DBLP:journals/pvldb/ChiangM08,DBLP:conf/icde/FanGLX09}. The confidence can be interpreted as an estimate of the probability that a randomly drawn pair of tuples agreeing on $X$ also agree on $Y$~\cite{DBLP:journals/tcs/KivinenM95,DBLP:journals/tods/CaldersNW02}.
Scheffer~\cite{DBLP:journals/ida/Scheffer05} study the trade off between support and confidence for finding association rules~\cite{DBLP:conf/sigmod/AgrawalIS93}, by computing a expected prediction accuracy.
In addition, Chiang and Miller~\cite{DBLP:journals/pvldb/ChiangM08}
also study some other measures such as conviction and $\chi^2$-test for evaluating dependency rules.
When a candidate $X\rightarrow Y$ is suggested together with minimum
support and confidence, Golab et
al.~\cite{DBLP:journals/pvldb/GolabKKSY08} study the discovery of
optimal \cfds\ with the minimum pattern tableau size. A concise set
of patterns are naturally desirable which may have lower cost during
the applications such as violation detection by \cfds. On the other hand,
Chiang and Miller~\cite{DBLP:journals/pvldb/ChiangM08} explore
\cfds\ by considering all the possible dependency candidates when
$X\rightarrow Y$ is not specified.
In~\cite{DBLP:conf/icde/FanGLX09}, Fan et al. also study the case
when the embedded \fds\ are not given, and propose three algorithms
for different scenarios.


The concept of matching dependencies (\mds) is first proposed
in~\cite{DBLP:conf/pods/Fan08} for specifying matching rules for the
object identification (see~\cite{DBLP:journals/tkde/ElmagarmidIV07}
for a survey). The \mds\ can be regarded as a generalization of
\fds, which are based on identical values having matching similarity
equal to $1.0$ exactly. Thus, \fds\ can be represented by the syntax of
\mds\ as well. For any two tuples, if their $X$ values are identical
(with similarity threshold $1.0$), then a \fd\ $(X\rightarrow Y)$
requires that their $Y$ values are identical too, i.e., a \md\
$(X\rightarrow Y, <1.0, 1.0>)$.
Koudas et al.~\cite{DBLP:conf/icde/KoudasSSV09} also study the dependencies with matching similarities on attributes $Y$ when given the \emph{exactly} matched values on $X$, which can be treated as a special case of \mds.
The reasoning mechanism for deducing \mds\ from a set of given \mds\ is studied in~\cite{DBLP:journals/pvldb/FanLJM09}. The \mds\ and their reason techniques can improve both the quality and efficiency of various record matching methods.

\section{Utility Measures}\label{sect_model}

In this section, we formally introduce the definitions of \mds. Then, we develop utility measures for
evaluating \mds\ over a given database instance.


Traditional functional dependencies \fds\ and their extensions rely on the exact matching operator $=$ to identify dependency relationships.
However, in the real world application, it is not possible to use exact matching operator
$=$ to identify matching over fuzzy data values such as text values.
For instance, $\mathsf{Jason~Smith}$ and $\mathsf{J. Smith}$ of attribute $\mathsf{Name}$ may refer to the same real world entity.
Therefore, instead of \fds\ on identical values, the \emph{matching dependencies} \mds~\cite{DBLP:conf/pods/Fan08} are proposed based on the matching quality. For text values, we can adopt the similarity matching operators, denoted by $\approx$, such as \emph{edit distance}~\cite{DBLP:journals/csur/Navarro01}, \emph{cosine similarity} with word tokens~\cite{DBLP:conf/sigmod/Cohen98} or \emph{q-grams}~\cite{DBLP:conf/www/GravanoIKS03}.

Consider a relation $\mathcal{R}(A_1,\dots,A_M)$ with $M$ attributes. Following similar syntax of \fds, we define \mds\ as following:
\footnote{The \mds\ syntax is described with two relation schema
$R_1, R_2$ for object identification in~\cite{DBLP:conf/pods/Fan08}, which can also be represented in a single relation schema $R$ as the \fds.}
\begin{definition}
\label{def_md_dependency}
A \emph{matching dependency (\md)} $\varphi$ is a pair $(X\rightarrow Y, \lambda)$, where $X\subseteq\mathcal{R}, Y\subseteq\mathcal{R}$, and $\lambda$ is a \emph{threshold pattern} of matching similarity thresholds on attributes in $X\cup Y$, e.g., $\lambda[A]$ denotes the matching similarity threshold on attribute $A$.
\end{definition}

A \md\ $\varphi$ specifies a constraint on the set of attributes $X$ to $Y$. Specifically, the constraint states that, for any two tuples $t_1$ and $t_2$ in a relation instance $r$ of $\mathcal{R}$, if $\bigwedge_{A_i\in X}t_1[A_i]\approx_{\lambda[A_i]} t_2[A_i]$, then $\bigwedge_{A_j\in Y}t_1[A_j]\approx_{\lambda[A_j]} t_2[A_j]$, where $\lambda[A_i]$ and $\lambda[A_j]$ are the \emph{matching similarity thresholds} on the attributes of $A_i$ and $A_j$ respectively.
In the above constraint, for each attribute $A_i\in X \cup Y$, the similarity matching operator $\approx$ indicates $\mathsf{true}$, if the similarity between $t_1[A_i]$ and $t_2[A_i]$ satisfies the corresponding threshold $\lambda[A_i]$. For example, a \md\ $\varphi([\mathsf{Street}]\rightarrow[\mathsf{City}], <0.8, 0.7>)$ in the
$\mathsf{Contacts}$ relation denotes that if two tuples has similar
$\mathsf{Street}$ (with matching similarity greater than $0.8$)
then their $\mathsf{City}$ values are probably similar as well (with
similarity at least $0.7$).

Like \fds\ and
\cfds~\cite{DBLP:journals/pvldb/GolabKKSY08, DBLP:journals/pvldb/ChiangM08}, we adopt \emph{support} and \emph{confidence} measures to evaluate the matching dependencies.
According to the above constraint of \mds, we need to consider the matching quality (e.g.,
cosine similarity or edit distance) of any pair of tuples $t_1$ and $t_2$ for $\mathcal{R}$.
Therefore, we compute a statistical distribution (denoted by $\mathcal{D}$) of the quality of pair-wised tuple matching for $\mathcal{R}$. The statistical distribution has a schema $\mathcal{D}(A_1,\dots,A_M, P)$, where each attribute $A_i$ in $\mathcal{D}$ corresponds to the matching quality values on the attribute $A_i$ of $\mathcal{R}$, and $P$ is the statistical value.
Let $s$ be a statistical tuple in $\mathcal{D}$. The statistic $s[P]$ denotes the probability that any two tuples $t_1$ and $t_2$ of $\mathcal{R}$ have the matching quality values $s[A_i]$, $\forall A_i\in \mathcal{R}$.
With a pair-wised evaluation of matching quality of all the $N$ tuples for $\mathcal{R}$, we can easily compute $P$ by $\frac{count(s)}{N*(N-1)/2}$, where $count(s)$ records the pairs of tuples having matching quality $s$.
Different matching operators have various spaces of matching values, such as cosine similarity in $[0.0,1.0]$ while edit distance having edit operations $1,2,\dots$.
In order to evaluate in a consistent environment, we map these matching quality values $s[A]$ to a unified space, say $[0,d-1]$, which is represented by $\mathsf{dom}(A)$ with $d$ elements.
Table~\ref{table_model_general} shows an example of the statistical
distribution $\mathcal{D}$ computed from $\mathsf{Contacts}$ in Table \ref{table_data_relation} by mapping\footnote{E.g., cosine similarity
value $s$ times $d-1$} the cosine similarities in $[0.0,1.0]$ to elements in $[0,d-1]$ of $\mathsf{dom}(A)$ with $d=10$. According to $\mathsf{dom}(A)$ in
our example, the first tuple $(1, 0, 3,\dots, 0.065)$ denotes that
there are about $6.5\%$ matching pairs in all pair-wised tuple matching, whose similarities are $1, 0, 3,\dots$ on the
attribute $A_1, A_2, A_3, \dots$ respectively.

 \begin{table}[h]
 \caption{Example of statistical distribution $\mathcal{D}$}
 \label{table_model_general}
 \centering
 \begin{tabular}{|cccccc|c|c}
 \cline{1-7}
 $A_1$ & $A_2$ & $A_3$ & $A_4$ & $A_5$ & $A_6$   & $P$ & \\ \cline{1-7}
 1 & 0 & 3 & 5 & 8                   & 4   & 0.065 & $s_1$ \\
 7 & 4 & 0 & 0 & 4                   & 1   & 0.043 & $s_2$ \\
 0 & 4 & 8 & 1 & 6                   & 2   & 0.124 & $s_3$ \\
 \vdots & \vdots   & \vdots & \vdots & \vdots & \vdots & \vdots & \vdots \\ \cline{1-7}
 \end{tabular}
 \end{table}

Then, we can measure the support and confidence of \mds, with various attributes $X$ and $Y$, based on the statistical distribution $\mathcal{D}$. Let $\lambda_X$ and $\lambda_Y$ be the projections of matching similarity threshold pattern $\lambda$ on the attributes of $X$ and
$Y$ respectively in a \md\ $\varphi$, which are also specified in terms
of elements in $\mathsf{dom}(A)$ of each $A\in X\cup Y$. Let $Z$ be the set of attributes not specified by $\varphi$, i.e., $\mathcal{R}\setminus (X\cup Y)$. The definitions of support and confidence for the
\md\ $\varphi(X\rightarrow Y, \lambda)$ are presented as follows:
\begin{eqnarray*}
\mathsf{support}(\varphi)&=&P(X\vDash\lambda_X, Y\vDash\lambda_Y)\\
&=&\sum_{Z}P(X\vDash\lambda_X, Y\vDash\lambda_Y, Z)\\
\label{equ_support_general}
\mathsf{confidence}(\varphi)&=&P(Y\vDash\lambda_Y \mid X\vDash\lambda_X)\\
&=&\frac{\sum_{Z}P(X\vDash\lambda_X, Y\vDash\lambda_Y,
Z)}{\sum_{Y, Z}P(X\vDash\lambda_X, Y, Z)}
\label{equ_confidence_general}
\end{eqnarray*}
where $\vDash$ denotes the \emph{satisfiability} relationship, i.e.,
$X\vDash\lambda_X$ denotes that the similarity values on all attributes in $X$ satisfy the corresponding thresholds listed in $\lambda_X$.
For example, we say that a statistical tuple $s$ in $\mathcal{D}$ satisfies $\lambda_X$, i.e.,
$s[X]\vDash\lambda_X$, if $s$ has similarity values higher than the corresponding minimum threshold, i.e., $s[A]\geq \lambda[A]$, for each attribute $A$ in $X$.

Consider any two tuples $t_1$ and $t_2$ from the original data
relation $\mathcal{R}$, the $\mathsf{support}(\varphi)$ estimates the
probability that the matching similarities of $t_1$ and $t_2$ on
attributes $X$ and $Y$ satisfy the thresholds specified by $\lambda_X$ and $\lambda_Y$, respectively. Similarly, the $\mathsf{confidence}(\varphi)$ computes the conditional probability that
the matching similarities between $t_1$ and $t_2$ on $Y$ satisfy the thresholds specified by
$\lambda_Y$ (i.e., $Y\vDash\lambda_Y$) given the condition that $t_1$ and $t_2$ are similar on attributes $X$ (i.e., $X\vDash\lambda_X$). Thus, high $\mathsf{confidence}(\varphi)$
means few instances of matching pairs that are similar on attributes
$X$ (i.e., $X\vDash\lambda_X$) but not similar on attributes $Y$
(i.e., $Y \nvDash\lambda_Y$), where $\nvDash$ denotes the
unsatisfiability relationship.

In real applications like inconsistency detection, in order to
achieve high detection accuracy, we would like to use \mds\
with high confidence. On the other hand, if users need high
recall of detection, then \mds\ with high support are preferred.
Intuitively, we
would like to discover those \mds\ with high support and high
confidence. Therefore, in the following of this paper, we study the problem of discovering \mds\ that can satisfy users minimum utility requirement of support $\eta_s$ and confidence $\eta_s$.


\section{Exact Algorithm}\label{sect_optimal_md}

We now study the determination of matching similarity threshold pattern for \mds\ based on the statistical distribution, which is a new problem different from \fds. In fact, once the $X\rightarrow Y$ is given for a \fd, it already implies the similarity threshold to be
$1.0$, that is, $(X\rightarrow Y, <1.0, 1.0>)$ if it is represented by
the \md\ syntax. Unlike \fds, we have various settings of matching
similarity thresholds for \mds. Therefore, in this section, we discuss how to find the right similarity thresholds in order to discover the \mds\ satisfying the required support and confidence.


\subsection{Problem Statement}\label{sect_optimal_discrete_confidence}

In order to discover a \md\ $\varphi$ with the minimum requirements of support $\eta_{s}$ and confidence $\eta_{c}$, the following preliminary should be given first:
\textbf{(I)} what is $Y$? and \textbf{(II)} what is matching quality requirement $\lambda_Y$. These two preliminary questions are usually addressed by specific applications. For example, if we would like to use discovered \mds\ to guide objet identification in the $\mathsf{Contacts}$ table, then $Y=\mathsf{SIN}$. The $\lambda_Y$ is often set to high similarity thresholds by applications to ensure high matching quality on $Y$ attributes. For example, $\lambda_Y$ is set to $1.0$ for $Y=\mathsf{SIN}$ in the object identification application.
Note that without the preliminary $\lambda_Y$, the discovered \mds\
will be meaningless. For example, a \md\ with $\lambda_Y=0$ can
always satisfy any requirement of $\eta_c, \eta_s$. Since all the
statistical tuples can satisfy the thresholds $\lambda_Y=0$, the
corresponding support and confidence will always be equal to $1.0$.

\begin{definition}
The threshold determination problem of \mds\ is: given the
minimum requirements of support and confidence $\eta_s, \eta_c$ and
the matching similarity threshold pattern $\lambda_Y$, find all the \mds\
$\varphi(X \rightarrow Y, \lambda)$ with threshold pattern $\lambda_X$ on attributes $X$ having
$\mathsf{confidence}(\varphi)\geq\eta_c$ and
$\mathsf{support}(\varphi)\geq\eta_s$, if exist; otherwise return
\emph{infeasible}.
\end{definition}

The attributes $X$ can be initially assigned by $\mathcal{R} \setminus Y$ if no suggestion is provided by specific applications, since our discovery process can automatically remove those attributes that are not required in $X$ for a \md\ $\varphi$.
Specifically, when a possible discovered threshold $\lambda[A]$ on
attribute $A$ is $0\in\mathsf{dom}(A)$, it means that any matching similarity value of the attribute $A \in X$ can satisfy the threshold $0$ and will not affect the \md\ $\varphi$ at all. In other words, the attribute $A$ can be removed from $X$ of the
\md\ $\varphi$.

\subsection{Exact Algorithm}\label{sect_optimal_discrete_exact}

Now, we present an algorithm to compute the matching similarity
thresholds on attributes $X$ for \mds\ having support and confidence greater than $\eta_s$ and $\eta_c$, respectively.
%
Let $A_1,\dots,A_{m_X}$ be the $m_X$ attributes in $X$. For simplicity, we use $\lambda$ to denote the threshold pattern projection $\lambda_X$
with $\lambda[A_1], \dots, \lambda[A_{m_X}]$ on all the $m_X$
attributes of $X$. Since, each threshold $\lambda[A]$ on attribute
$A$ is a value from $\mathsf{dom}(A)$, i.e.,
$\lambda[A]\in\mathsf{dom}(A)$, we can investigate all the
possible candidates of threshold pattern $\lambda$. Let $\mathcal{C}_t$ be the
set of all the possible threshold pattern candidates, having
$$\mathcal{C}_t=\mathsf{dom}(A_1)\times\dots\times\mathsf{dom}(A_{m_X})=\mathsf{dom}(X).$$
The total number of candidates is
$c=|\mathcal{C}_t|=|\mathsf{dom}(X)|=d^m$, where $d$ is the size of
$\mathsf{dom}(A)$.

Let $n$ be the number of statistical tuples in the input
statistical distribution $\mathcal{D}$. We consider two
statistical values $P_i^{j}(X, Y)$ and $P_i^{j}(X)$, which record
$P(X\vDash\lambda_X, Y\vDash\lambda_Y)$ and
$P(X\vDash\lambda_X)$ respectively for the candidate
$\lambda_j\in\mathcal{C}_t$ based on the information of the first $i$
tuples in $\mathcal{D}$, initially having $P_0^{j}(X, Y)=P_0^{j}(X)=0$.
The recursion is defined as follows, with $i$ increasing from $1$ to
$n$ and $j$ increasing from $1$ to $c$.
\begin{eqnarray*}
P_i^{j}(X, Y) &=&  \begin{cases}
   P_{i-1}^{j}(X, Y)+s_i[P], & \mathrm{if~} s_i[X]\vDash\lambda_j, s_i[Y]\vDash\lambda_Y \\
   P_{i-1}^{j}(X, Y), & \mathrm{otherwise}
   \end{cases}\\
P_{i}^{j}(X) &=&  \begin{cases}
   P_{i-1}^{j}(X)+s_i[P], & \mathrm{if~} s_i[X]\vDash\lambda_j  \\
   P_{i-1}^{j}(X), & \mathrm{otherwise}
   \end{cases}
\end{eqnarray*}
Finally, those $\lambda_j$ can be returned if
$\mathsf{support}=P_n^{j}\geq \eta_s$ and
$\mathsf{confidence}=\frac{P_n^{j}(X, Y)}{P_n^{j}(X)}\geq \eta_c$.

 \begin{algorithm}[h]
 \caption{Exact algorithm \textbf{EA}($\mathcal{D}, \mathcal{C}_t$)} %
 \label{alg_md_exact} %
 \begin{algorithmic}[1]
    \FOR{each candidate $\lambda_j\in \mathcal{C}_t, j:1\rightarrow c$}
     \STATE $P_0^{j}(X, Y)=P_0^{j}(X)=0$
     \FOR{each statistical tuples $s_i\in \mathcal{D}, i:1\rightarrow n$}
      \STATE compute $P_i^{j}(X, Y),P_{i}^{j}(X)$
     \ENDFOR
    \ENDFOR
    \RETURN $\lambda_j$ with confidence and support satisfying $\eta_c,\eta_s$
 \end{algorithmic}
 \end{algorithm}

We can implement the exact algorithm (namely \ea) by considering all
the statistical tuples $s_i$ in $\mathcal{D}$ with $i$ from $1$ to
$n$, whose time complexity is $\mathcal{O}(nc)$.

\subsection{Pruning Strategies}\label{sect_optimal_discrete_pruning}

Since the original exact algorithm needs to traverse all the $n$ statistical
tuples in $\mathcal{D}$ and $c$ candidate threshold patterns in
$\mathcal{C}_t$, which is very costly. In fact, with the given $\eta_s$ and $\eta_c$, we can investigate the relationship between similarity thresholds and avoid checking all candidate threshold patterns in $\mathcal{C}_t$ and all statistical tuples in $\mathcal{D}$. Therefore, in the following two subsections, we present pruning techniques based on the given support and confidence, respectively.

\paragraph{Pruning by support}
We first study the relationships among different threshold patterns, based on which we then propose rules to filter out candidates that have supports lower than $\eta_s$.

\begin{definition}\label{def_md_dominate}
Given two similarity threshold patterns $\lambda_1$ and $\lambda_2$, if $\lambda_1[A]\leq
\lambda_2[A]$ holds for all the attributes, $\forall A \in X$, then $\lambda_1$ \emph{dominates} $\lambda_2$,
denoted as $\lambda_1\lessdot \lambda_2$.
\end{definition}


Based on the \emph{dominate} definition, the following Lemma describes the relationships of supports between similarity threshold patterns.

\begin{lemma}\label{lemma_md_support_increase}
Given two \mds, $\varphi_1=(X \rightarrow Y, \lambda_1)$ and $\varphi_2=(X \rightarrow Y, \lambda_2)$ over the same relation instance of $\mathcal{R}$, if $\lambda_1$ dominates $\lambda_2$, $\lambda_1\lessdot \lambda_2$, then we have
$\mathsf{support}(\varphi_1)\geq \mathsf{support}(\varphi_2)$.
\end{lemma}
\begin{proof}
Let $\mathsf{cover}(\lambda_1)$ and $\mathsf{cover}(\lambda_2)$
denote the set of statistical tuples that satisfy the threshold
$\lambda_1$ and $\lambda_2$ respectively, e.g.,
$\mathsf{cover}(\lambda_2)=\{s\mid s[X]\vDash\lambda_2, s\in
\mathcal{D}\}$. According to the minimum similarity thresholds, for
each attribute $A$, we have $\lambda_2[A]\leq s[A]$. In
addition, since $\lambda_1\lessdot\lambda_2$, for any tuple $s\in
\mathsf{cover}(\lambda_2)$, we also have
$\lambda_1[A]\leq\lambda_2[A]\leq s[A]$ on all the attributes
$A$. In other words, the set of statistical tuples covered by
$\lambda_2$ also satisfy the threshold of $\lambda_1$, i.e.,
$\mathsf{cover}(\lambda_2)\subseteq\mathsf{cover}(\lambda_1)$.
Referring to the definition of $\mathsf{support}$, we have
$\mathsf{support}(\varphi_1)\geq \mathsf{support}(\varphi_2)$.
\end{proof}


According to Lemma \ref{lemma_md_support_increase}, given a candidate similarity threshold pattern $\lambda_j$ having support lower than the user specified requirement $\eta_s$, i.e., $P_n^{j}(X, Y)<\eta_s$,  all the candidates that are dominated by $\lambda_j$ should have support lower than $\eta_s$ and can be safely pruned without computing their associated support and confidence.


We present the implementation of pruning by support (namely \eps) in
Algorithm~\ref{alg_md_prune_support}.

\begin{algorithm}[h]
\caption{Pruning by support \textbf{EPS}($\mathcal{D},\mathcal{C}_t$)} %
\label{alg_md_prune_support} %
\begin{algorithmic}[1]
   \FOR{each candidate $\lambda_j\in \mathcal{C}_t, j:1\rightarrow c$}
    \STATE $Q^a_0{j}=Q^b_0{j}=0$
    \FOR{each tuple $s_i\in \mathcal{D}, i:1\rightarrow n$}
     \STATE compute $Q^a_i{j},P_{i}^{j}(X)$
    \ENDFOR
    \IF{$Q^a_n{j}<\eta_s$}
     \STATE remove all the remaining candidates $\lambda'$ dominated by $\lambda_j$ from
     $\mathcal{C}_t$ \hfill \COMMENT{Pruning by support, $\lambda'\gtrdot \lambda_j$}
    \ENDIF
   \ENDFOR
   \RETURN $\lambda_j$ with confidence and support satisfying $\eta_c,\eta_s$
\end{algorithmic}
\end{algorithm}

In order to maximize the pruning, we can heuristically select an
ordering of candidates in $\mathcal{C}_t$ that for any $j_1 < j_2$
having $\lambda_{j_1}\lessdot \lambda_{j_2}$. That is, we always
first process the candidates that dominate others. In fact, we can use a DAG (directed acyclic graph), $\mathcal{G}$, to represent candidate similarity patterns as vertices and dominant relationships among the similarity patterns as edges. Therefore, the dominant order of candidate patterns can be obtained by a \textsc{bfs} traversal upon $\mathcal{G}$.

\paragraph{Pruning by confidence}
Other than pruning by support, we can also utilize the given confidence requirement to avoid further examining tuples that have no improvement of
confidence when the confidence is already lower than $\eta_c$ for a
candidate $\lambda_j$.

We first group the statistical tuples in $\mathcal{D}$
into two parts based on the preliminary $\lambda_Y$ as follows. Let
$k$ be a pivot between $1$ and $n$. For the first $k$ tuples, we
have $s_i[Y]\vDash\lambda_Y, 1\leq i\leq k$. All the remaining
$n-k$ tuples have $s_i[Y]\nvDash\lambda_Y, k+1 \leq i\leq n$.
This grouping of statistical tuples in $\mathcal{D}$ can be done in linear time.

\begin{lemma}\label{lemma_md_confidence_decrease}
Consider a pre-grouped statistical distribution $\mathcal{D}$. For
any $1\leq i_1<i_2\leq n$, we always have
$$\frac{P_{i_1}^{j}(X, Y)}{P_{i_1}^{j}(X)}\geq\frac{P_{i_2}^{j}(X, Y)}{P_{i_2}^{j}(X)}.$$
\end{lemma}
\begin{proof}
Since the first $k$ tuples have $s_i[Y]\vDash\lambda_Y$, according
to the computation of $P(X, Y)$ and $P(X)$, we have
$$\frac{P_{i}^{j}(X, Y)}{P_{i}^{j}(X)}=1.0, \quad 1\leq i\leq k.$$
Moreover, for
the remaining $n-k$ tuples with $s_i[Y]\nvDash\lambda_Y$, the
$P(X, Y)$ value will not change any more, i.e., $P_{i}^{j}(X, Y)=P_{k}^{j}(X,Y),
k+1 \leq i \leq n$. Meanwhile, the corresponding $P(X)$ is
non-decreasing, that is, $P_{k}^{j}(X)\leq P_{i_1}^{j}(X) \leq
P_{i_2}^{j}(X)$ for any $k+1\leq i_1<i_2 \leq n$. Consequently, we
have
$$\frac{P_{i_1}^{j}(X, Y)}{P_{i_1}^{j}(X)}\geq\frac{P_{i_2}^{j}(X, Y)}{P_{i_2}^{j}(X)},\quad k+1\leq i_1<i_2 \leq n.$$
Combining above two statements, we proved the lemma.
\end{proof}

Therefore, according to the formula of confidence, with the increase
of $i$ from $1$ to $n$, the confidence of a specific candidate
$\lambda_j$ is non-increasing. For a candidate
$\lambda_j$, when processing the statistical tuple $s_i$, if the
current confidence $\frac{P_i^{j}(X, Y)}{P_{i}^{j}(X)}$ is lower than
$\eta_c$, then we can prune the candidate $\lambda_j$ without
considering the remaining statistical tuples from $i+1$ to $n$ in
$\mathcal{D}$.

\begin{algorithm}[h]
\caption{Pruning by support \& confidence \textbf{EPSC}($\mathcal{D},\mathcal{C}_t$)} %
\label{alg_md_prune_both} %
\begin{algorithmic}[1]
   \FOR{each candidate $\lambda_j\in \mathcal{C}_t, j:1\rightarrow c$}
    \STATE $P_0^{j}(X, Y)=P_0^{j}(X)=0$
    \FOR{each tuple $s_i\in \mathcal{D}, i:1\rightarrow n$}
     \STATE compute $P_i^{j}(X,Y),P_{i}^{j}(X)$
     \IF{$\frac{P_{i}^{j}(X, Y)}{P_{i}^{j}(X)}<\eta_c$}
      \STATE remove $\lambda_j$ from $\mathcal{C}_t$ \hfill \COMMENT{Pruning by confidence}
      \IF{$P_{i}^{j}(X, Y)\geq\eta_s$}
       \STATE \textbf{break}
      \ENDIF
     \ENDIF
    \ENDFOR
    \IF{$P_n^{j}(X, Y) <\eta_s$}
     \STATE remove all the remaining candidates $\lambda'$ dominated by $\lambda_j$ from
     $\mathcal{C}_t$ \hfill \COMMENT{Pruning by support, $\lambda'\gtrdot \lambda_j$}
    \ENDIF
   \ENDFOR
   \RETURN $\lambda_j$ with confidence and support satisfying $\eta_c,\eta_s$
\end{algorithmic}
\end{algorithm}

Finally, both the pruning by support and the pruning by confidence are cooperated
together into a single threshold determination algorithm as shown in Algorithm~\ref{alg_md_prune_both}(namely \epsc). We also demonstrate the performance of the hybrid pruning \epsc\ in Section~\ref{sect_experiment}.

\section{Approximation Algorithm}\label{sect_approximation}

Though we have proposed pruning rules for exact method (Algorithm \ref{alg_md_prune_both}), the whole evaluation space is still all the $n$ tuples in statistical distribution $\mathcal{D}$. Therefore, in this section, we present an approximate algorithm which only traverses the first $k$ ($k=1, \ldots n$) tuples in $\mathcal{D}$, with bounded relative errors on support and confidence of returned \mds.

Let $C^n$ and $S^n$ be the confidence and support computed in the exact
solution with all $n$ tuples. We study the approximate confidence and support, $C^k$ and $S^k$,
by ignoring the statistical tuples from $s_{k+1}$ to $s_n$. For a candidate threshold pattern $\lambda_j\in\mathcal{C}_t$, let
$$ \beta = P_k^{j}(X), \quad \bar{\beta}  = P_n^{j}(X) - P_k^{j}(X),$$
where $\beta$ denotes $P(X\vDash\lambda_X)$ for the candidate
$\lambda_j$ based on the first $k$ tuples in $\mathcal{D}$, and
$\bar{\beta}$ is $P(X\vDash\lambda_X)$ based on the remaining
$n-k$ tuples.
The following Lemma indicates the error bounds of $C^k$ and $S^k$ when $\bar{\beta}$ for a specific $k$ is in a certain range.

\begin{lemma}\label{lemma_md_appx_bound}
If we have $\bar{\beta}\leq \min (\epsilon\eta_s,
\frac{\epsilon\eta_s\eta_c}{1-\epsilon-\eta_c})$, then the error of
approximate confidence $C^k$ compared to the exact confidence
$C^n$ is bounded by $-\epsilon\leq\frac{C^n-C^k}{C^n}\leq\epsilon$,
and the error of approximate support $S^k$ compared to the exact $S^n$ is bounded by $\frac{S^n-S^k}{S^n}\leq\epsilon$.
\end{lemma}
\begin{proof}
Let \begin{eqnarray*}
\alpha &=& P_k^{j}(X, Y) \\
\bar{\alpha}  &=& P_n^{j}(X, Y) - P_k^{j}(X, Y)\end{eqnarray*}
According to the computation of confidence, we have
$C^k=\frac{\alpha}{\beta}$ and
$C^n=\frac{\alpha+\bar{\alpha}}{\beta+\bar{\beta}}$. Let
$Z=1-\frac{C^n-C^k}{C^n}=\frac{C^k}{C^n}$, that is,
\begin{equation*}
Z=\frac{\alpha(\beta+\bar{\beta})}{\beta(\alpha+\bar{\alpha})} \leq
1+\frac{\bar{\beta}}{\beta}
\end{equation*}

First, we have $\beta=\alpha+\sum_{i=1}^k
s_i[P(X\vDash\lambda_j,Y\nvDash\lambda_Y)]\geq\alpha$. Note that
$\alpha$ is the approximate support of the \md\ $\varphi$ with
matching similarity threshold pattern $\lambda_j$ on the attributes $X$.
According to the minimum support constraint, for a valid
$\lambda_j$, we have $\beta\geq\alpha\geq\eta_s$. Thereby,
$$Z\leq1+\frac{\bar{\beta}}{\eta_s}$$
Moreover, according to the condition $\bar{\beta}\leq \min
(\epsilon\eta_s, \frac{\epsilon\eta_s\eta_c}{1-\epsilon-\eta_c})$,
that is $\bar{\beta}\leq\epsilon\eta_s$, we have
$$Z\leq1+\epsilon$$

Second, similar to $\beta\geq\alpha$, we also have
$\bar{\alpha}\leq\bar{\beta}$ for the tuples from $k+1$ to $n$.
Therefore,
$$Z\geq\frac{\alpha(\beta+\bar{\beta})}{\beta(\alpha+\bar{\beta})}= \frac{\beta+\bar{\beta}}{\beta+\frac{\beta\bar{\beta}}{\alpha}}$$
According to the minimum confidence
$\frac{\alpha}{\beta}\geq\eta_c$,
\begin{equation}\label{equ_md_error_bound} Z\geq\frac{\beta+\bar{\beta}}{\beta+\frac{\bar{\beta}}{\eta_c}}
=1-\frac{\bar{\beta}(1-\eta_c)}{\beta\eta_c+\bar{\beta}}\end{equation}
Recall that $\beta\geq\eta_s$ and the confidence should be lower
than or equal to $1$, i.e., $\eta_c\leq1$. Thus,
$$Z\geq1-\frac{\bar{\beta}(1-\eta_c)}{\eta_s\eta_c+\bar{\beta}}=1-\frac{1-\eta_c}{\frac{\eta_c\eta_s}{\bar{\beta}}+1}$$
Since we have the condition
$\bar{\beta}\leq\frac{\epsilon\eta_s\eta_c}{1-\epsilon-\eta_c}$,
$$Z\geq1-\frac{1-\eta_c}{\frac{1-\epsilon-\eta_c}{\epsilon}+1}=1-\epsilon$$

Finally, based on the above two conditions, we conclude that
$$1+\epsilon\geq Z=1-\frac{C^n-C^k}{C^n}=\frac{C^k}{C^n}\geq 1-\epsilon$$
$$-\epsilon\leq \frac{C^n-C^k}{C^n} \leq \epsilon$$

On the other hand, according to the computation of support, we have
$S^k=\alpha$ and $S^n=\alpha+\bar{\alpha}$. Therefore,
$$\frac{S^n-S^k}{S^n} = \frac{1}{1+\frac{\alpha}{\bar{\alpha}}}$$
Recall that we have $\alpha\geq\eta_s$ and
$\bar{\alpha}\leq\bar{\beta}\leq\epsilon\eta_s$.
$$\frac{S^n-S^k}{S^n}\leq\frac{1}{1+\frac{1}{\epsilon}}=\frac{\epsilon}{1+\epsilon}<\epsilon$$
That is, the worst-case relative error is bounded by $\epsilon$ for
both the confidence and support.
\end{proof}

Now, we consider the last $n-k$ tuples in $\mathcal{D}$.
Let $$\bar{B}(k)=\sum_{i=k+1}^n s_i[P],$$
where $s_i[P]$ is the probability associated to each statistical tuple in
$\mathcal{D}$. Referring to the definition of $\bar{\beta}$,
for any $\lambda_j$, we always have $\bar{\beta}\leq\bar{B}(k)$. If
there exists a $k$ having $\bar{B}(k)\leq \min (\epsilon\eta_s,
\frac{\epsilon\eta_s\eta_c}{1-\epsilon-\eta_c})$, then
$\bar{\beta}\leq \min (\epsilon\eta_s,
\frac{\epsilon\eta_s\eta_c}{1-\epsilon-\eta_c})$ is satisfied for
all the threshold candidates $\lambda_j$. Since the $\bar{B}(k)$
decreases with the increase of $k$, to determine a minimum $k$ is to
find a corresponding maximum $\bar{B}(k)$.
Therefore, according to Lemma~\ref{lemma_md_appx_bound}, given an
error bound $\epsilon, 0<\epsilon<1-\eta_c$, we can compute a
minimum position $k=\arg\max_{k=1}^n \bar{B}(k)$ having
$\bar{B}(k)\leq \min (\epsilon\eta_s,
\frac{\epsilon\eta_s\eta_c}{1-\epsilon-\eta_c})$.

\begin{theorem}\label{the_approx}
Given an error bound $\epsilon, 0<\epsilon<1-\eta_c$, we can
determine a minimum $k$, having
$$\bar{B}(k)\leq \min
(\epsilon\eta_s, \frac{\epsilon\eta_s\eta_c}{1-\epsilon-\eta_c}),
1\leq k\leq n .$$ The approximation by considering first $k$
tuples in $\mathcal{D}$ finds approximate \mds\ with the error
bound $\epsilon$ on both the confidence and support compared with
the exact one. The complexity is $\mathcal{O}(kc)$.
\end{theorem}

Finally, we present the approximation implementation in
Algorithm~\ref{alg_md_approx}. Let $\bar{B}$ denotes
$\bar{B}(k)=\sum_{i=k+1}^n s_i[P]$ for the current $k$. With $k$
decreasing from $n$ to $1$, we can determine a minimum $k$ where
$\bar{B}=\bar{B}(k)\leq\min (\epsilon\eta_s,
\frac{\epsilon\eta_s\eta_c}{1-\epsilon-\eta_c})$ is still satisfied.
After computing $k$, we process the tuples $s_i$ starting from
$i=1$. When the bound condition is first satisfied, i.e., $i=k$ with
$\bar{B}=\bar{B}(k)\leq\min (\epsilon\eta_s,
\frac{\epsilon\eta_s\eta_c}{1-\epsilon-\eta_c})$, the processing
terminates. Here, the error bound $\epsilon$ is specified by user
requirement with $0<\epsilon<1-\eta_c$.

 \begin{algorithm}[h]
 \caption{Approximation algorithm \textbf{AP}($\mathcal{D},\mathcal{C}_t$)} %
 \label{alg_md_approx} %
 \begin{algorithmic}[1]
    \FOR{each tuple $s_k\in \mathcal{D}, k:n\rightarrow 1$}
      \STATE $\bar{B}$ \texttt{+=} $s_k[P]$
      \IF{$\bar{B}>\min (\epsilon\eta_s,
 \frac{\epsilon\eta_s\eta_c}{1-\epsilon-\eta_c})$}
      \STATE $k$\texttt{++}; \textbf{break} \hfill \COMMENT{Compute $k$}
     \ENDIF
    \ENDFOR
    \FOR{each candidate $\lambda_j\in \mathcal{C}_t, j:1\rightarrow c$}
     \STATE $P_{0}^{j}(X, Y)=P_{0}^{j}(X)=0$
     \FOR{each tuple $s_i\in \mathcal{D}, i:1\rightarrow k$}
      \STATE compute $P_{i}^{j}(X, Y),P_{i}^{j}(X)$
     \ENDFOR
    \ENDFOR
    \RETURN $\lambda_j$ with confidence and support satisfying $\eta_c,\eta_s$
 \end{algorithmic}
 \end{algorithm}

Given an error bound $\epsilon$, the bound condition is then fixed.
In order to minimize $k$, we expect that the $P$ values of the
tuples from $k+1$ to $n$ in $\bar{B}(k)=\sum_{j=k+1}^n s_j[P]$ are
small. In other words, an instance of $\mathcal{D}$ with higher $P$
in the tuples from $1$ to $k$ is preferred.
Therefore, we can reorganize the tuples in $\mathcal{D}$ in the
decreasing order of $P$ as the input of
Algorithm~\ref{alg_md_approx}. The ordering of statistical tuples in
$\mathcal{D}$ by the $P$ values can be done in linear time by
amortizing the $P$ values into a constant domain.

\paragraph{Approximation Individually}
We study the approximation by each individual candidate $\lambda_j$
with a more efficient bound condition respectively.
According to formula (\ref{equ_md_error_bound}) in the proof of
error bound, we find that for each specific candidate $\lambda_j$ if
$\bar{\beta}\leq\min(\epsilon\beta,
\frac{\epsilon\beta\eta_c}{1-\epsilon-\eta_c})$, then the error
bound is already satisfied and the processing can be terminated for
this $\lambda_j$. Therefore, rather than one fixed bound condition
for all the candidates, the bound of $\bar{\beta}$ can be determined
dynamically for each candidate $\lambda_j$ respectively during the
processing. Algorithm~\ref{alg_md_approx_prune} shows the
implementation of approximation with dynamic bound condition on each
candidate $\lambda_j$ individually.

\begin{algorithm}[h]
\caption{Approximation individually \textbf{API}($\mathcal{D},\mathcal{C}_t$)} %
\label{alg_md_approx_prune} %
\begin{algorithmic}[1]
  \FOR{each tuple $s_i\in \mathcal{D},i:n\rightarrow 1$}
     \STATE $\bar{B}$ \texttt{+=} $s_i[P]$
     \IF{$\bar{B}\leq\min (\epsilon\eta_s,
\frac{\epsilon\eta_s\eta_c}{1-\epsilon-\eta_c})$}
     \STATE $k=i$ \hfill \COMMENT{Compute $k$}
    \ENDIF
   \ENDFOR
   \FOR{each candidate $\lambda_j\in \mathcal{C}_t, j:1\rightarrow c$}
    \STATE $P_{0}^{j}(X, Y)=P_{0}^{j}(X)=0$
    \STATE $\bar{B}_j=\bar{B}$
    \FOR{each tuple $s_i\in \mathcal{D}, i:1\rightarrow k$}
     \STATE compute $P_{i}^{j}(X, Y),P_{i}^{j}(X)$
     \STATE $\beta =P_{i}^{j}(X)$
     \STATE $\bar{B}_j$ \texttt{-=} $s_i[P]$
     \IF{$\bar{B}_j\leq\min (\epsilon\beta,\frac{\epsilon\beta\eta_c}{1-\epsilon-\eta_c})$}
      \STATE \textbf{break}
     \ENDIF
    \ENDFOR
   \ENDFOR
   \RETURN $\lambda_j$ with confidence and support satisfying $\eta_c,\eta_s$
\end{algorithmic}
\end{algorithm}

\begin{corollary}
The worst case complexity of the approximation individually is
$\mathcal{O}(kc)$
\end{corollary}
\begin{proof}
Note that with the increasing of $i$ from $1$ to $k$, for a specific
$\lambda_j$, the value $\beta$ increases and $\bar{B}_j$ decreases.
For any $i< k$, if $\beta<\eta_s$, i.e., $\lambda_j$ is invalid
currently, the bound condition cannot be satisfied having
$$\min
(\epsilon\beta, \frac{\epsilon\beta\eta_c}{1-\epsilon-\eta_c})<\min
(\epsilon\eta_s,
\frac{\epsilon\eta_s\eta_c}{1-\epsilon-\eta_c})<\bar{B}_j .$$ When
$\lambda_j$ has $\beta\geq\eta_s$ as a valid threshold, the bound
condition is relaxed from $\min (\epsilon\eta_s,
\frac{\epsilon\eta_s\eta_c}{1-\epsilon-\eta_c})$ to $\min
(\epsilon\beta, \frac{\epsilon\beta\eta_c}{1-\epsilon-\eta_c})$.
Thereby, the bound condition may be satisfied by a smaller $i$ than
$k$, i.e.,
$$\min (\epsilon\eta_s,
\frac{\epsilon\eta_s\eta_c}{1-\epsilon-\eta_c})<\bar{B}_j\leq\min
(\epsilon\beta, \frac{\epsilon\beta\eta_c}{1-\epsilon-\eta_c}).$$
The
worst case is that all candidates do not achieve their bounds until
processing the tuple $s_k$, where
$$\bar{B}_j = \bar{B}(k) \leq \min
(\epsilon\eta_s, \frac{\epsilon\eta_s\eta_c}{1-\epsilon-\eta_c})
\leq \min (\epsilon\beta,
\frac{\epsilon\beta\eta_c}{1-\epsilon-\eta_c})$$
must be satisfied.
This is exact the Algorithm~\ref{alg_md_approx} without individual
approximation.
\end{proof}

Finally, we cooperate the pruning by support together with the
approximation (namely \aps) and the approximation individually
(namely \apsi) respectively. As we presented in the experimental
evaluation, the approximation techniques can further improve the
discovering efficiency with an approximate solution very close to
the exact one (bounded by $\epsilon$).


\begin{figure*}[t]
\begin{minipage}[b]{0.33\linewidth}
    \centering
    \includegraphics[width=\linewidth]{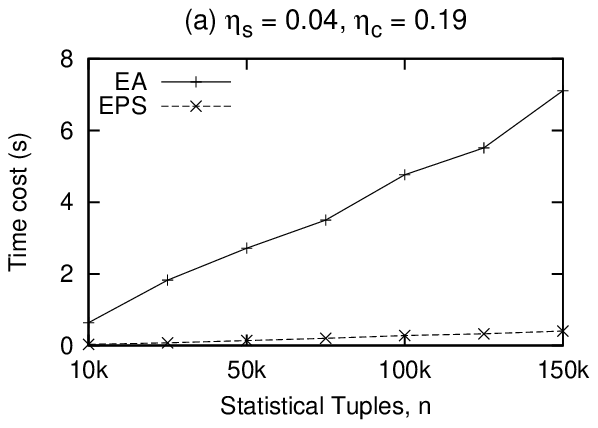}
    \includegraphics[width=\linewidth]{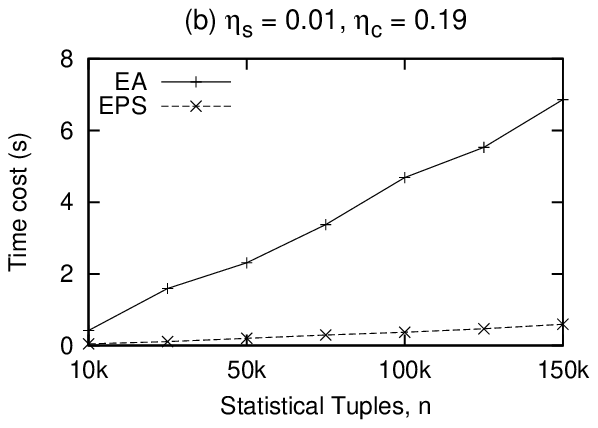}
    \caption{Pruning on \emph{CiteSeer} }
    \label{exp_eps_time_cite}
\end{minipage}%
\begin{minipage}[b]{0.33\linewidth}
    \centering
    \includegraphics[width=\linewidth]{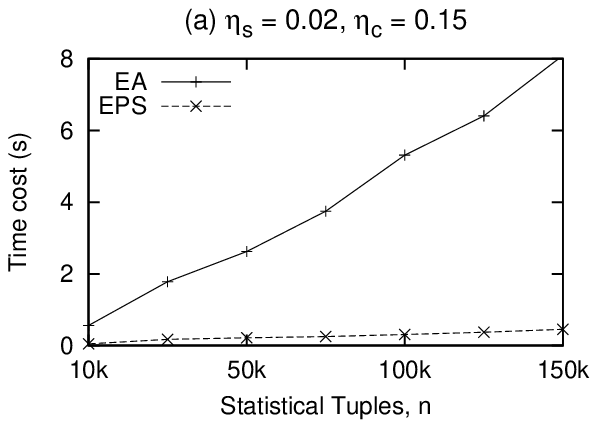}
    \includegraphics[width=\linewidth]{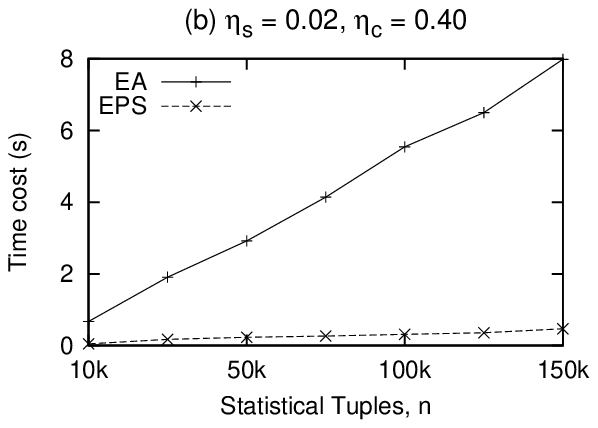}
    \caption{Pruning on \emph{Cora}}
    \label{exp_eps_time_cora}
\end{minipage}%
\begin{minipage}[b]{0.33\linewidth}
    \centering
    \includegraphics[width=\linewidth]{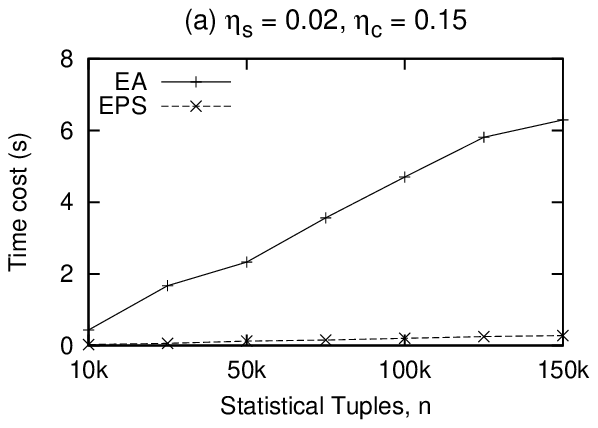}
    \includegraphics[width=\linewidth]{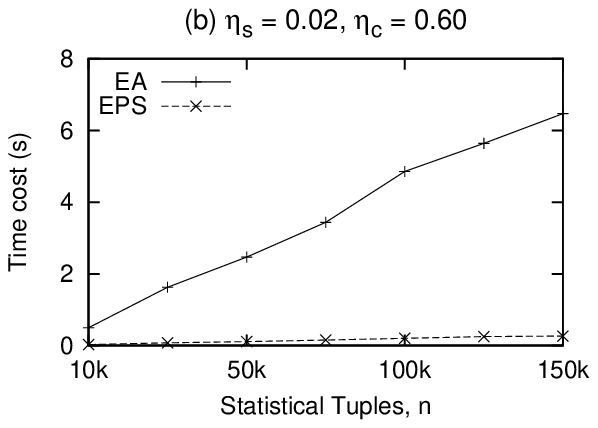}
    \caption{Pruning on \emph{Restaurant}}
    \label{exp_eps_time_rest}
\end{minipage}
\end{figure*}

\begin{figure*}[t]
\begin{minipage}[b]{0.33\linewidth}
    \centering
    \includegraphics[width=\linewidth]{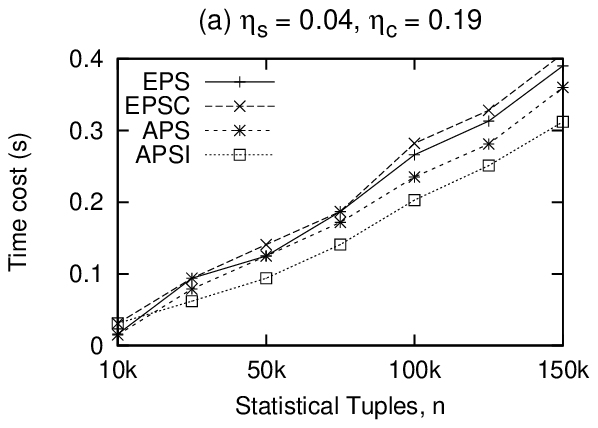}
    \includegraphics[width=\linewidth]{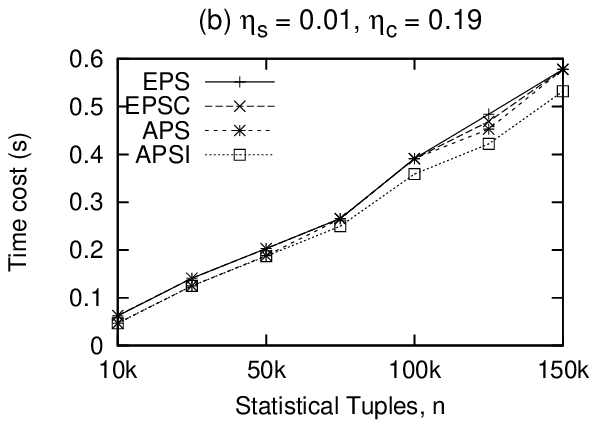}
    \caption{Advanced on \emph{CiteSeer} }
    \label{exp_advanced_time_cite}
\end{minipage}%
\begin{minipage}[b]{0.33\linewidth}
    \centering
    \includegraphics[width=\linewidth]{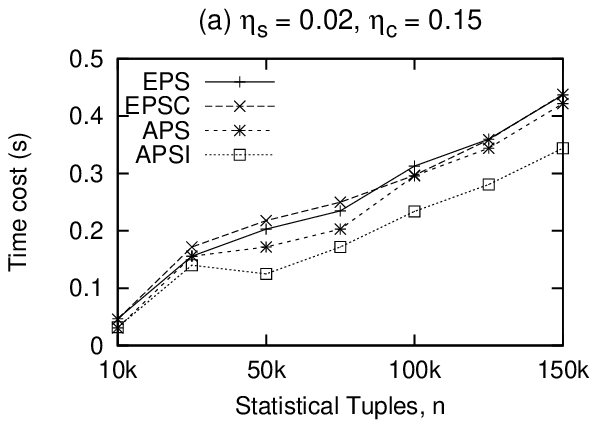}
    \includegraphics[width=\linewidth]{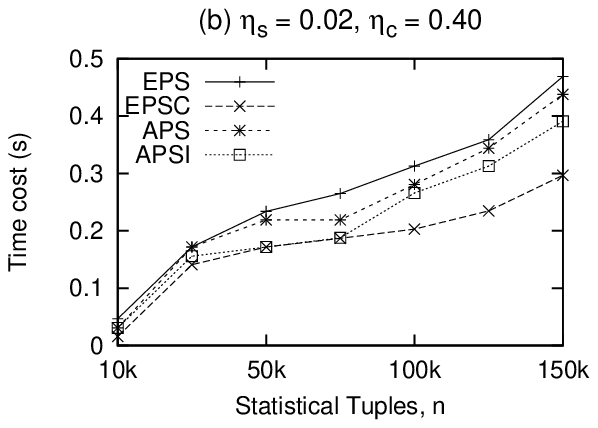}
    \caption{Advanced on \emph{Cora}}
    \label{exp_advanced_time_cora}
\end{minipage}%
\begin{minipage}[b]{0.33\linewidth}
    \centering
    \includegraphics[width=\linewidth]{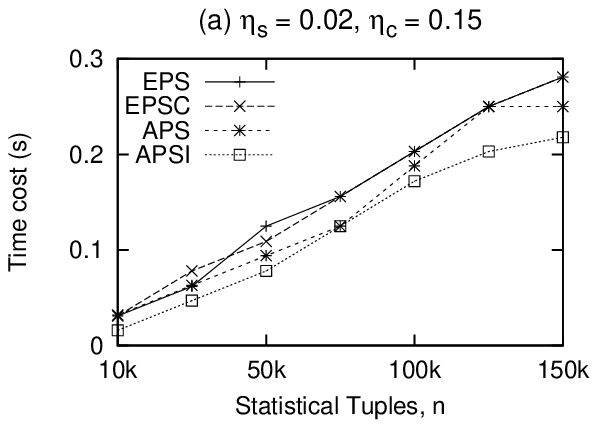}
    \includegraphics[width=\linewidth]{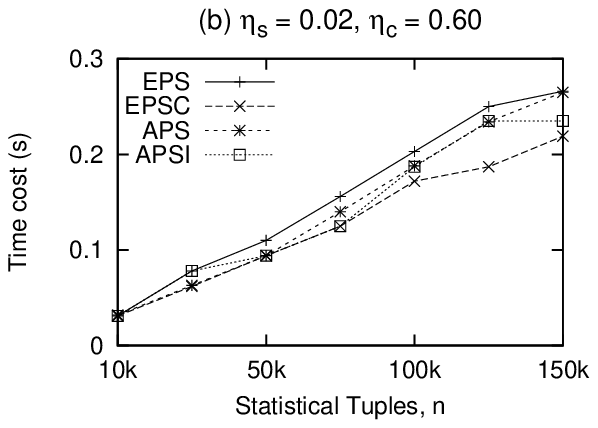}
    \caption{Advanced on \emph{Restaurant}}
    \label{exp_advanced_time_rest}
\end{minipage}
\end{figure*}

\section{Experimental Evaluation}\label{sect_experiment}

Now, we report the experiment evaluation on proposed methods.
All the algorithms are implemented by Java. The experiment evaluates
on a machine with Intel Core 2 CPU (2.13 GHz) and 2~GB of memory.

\paragraph{Experiment Setting}
%

In the experimental evaluation, we use three real data sets. The
\emph{Cora}\footnote{http://www.cs.umass.edu/\texttildelow mccallum/code-data.html}
data set, prepared by McCallum et
al.~\cite{DBLP:conf/kdd/McCallumNU00}, consists of 12 attributes
including $\mathsf{author, volume, title, institution, }$
$\mathsf{venue}$, etc.
The \emph{Restaurant}\footnote{http://www.cs.utexas.edu/users/ml/riddle/data.html} data set consists of restaurant records including attributes $\mathsf{name, address, city}$ and $\mathsf{type}$.
The \emph{CiteSeer}\footnote{http://citeseer.ist.psu.edu/} data set
is selected with attributes including $\mathsf{title, author,
address, affiliation, subject, description}$, etc.
We use the \emph{cosine} similarity to evaluate the matching quality of the
tuples in the original data.
By applying the $\mathsf{dom}(A)$ mapping in Section~\ref{sect_model}, we can obtain
statistical distributions with at most $186,031$ statistical tuples
in \emph{Cora}, $140,781$ statistical tuples
in \emph{Restaurant} and $314,382$ statistical tuples in \emph{CiteSeer}.
Our experimental evaluation is then conducted in several
pre-processed statistical distributions with various sizes of
statistical tuples $n$ from $10,000$ to $150,000$ respectively.

We mainly observes the efficiency of proposed algorithms. Since our main task
is to discover \mds\ under the required $\eta_s$ and
$\eta_c$, we study the runtime performance in various distributions
with different $\eta_s$ and $\eta_c$ settings. The discovery
algorithms determine the matching similarity settings of attributes
for \mds. Suppose that users want to discover \mds\ on the following $X\rightarrow Y$ of three data sets respectively: \textbf{i)} the dependencies on
$$Cora: \mathsf{author, volume, title\rightarrow venue}$$
with the preliminary requirement of
minimum similarity $0.6$ on $\mathsf{venue}$; \textbf{ii)} the dependencies on
$$Restaurant: \mathsf{name, address, type \rightarrow city}$$
with the preliminary requirement of
minimum similarity $0.5$ on $\mathsf{city}$;
and \textbf{iii)} the dependencies on
$$CiteSeer: \mathsf{address, affiliation, description\rightarrow subject}$$
with preliminary $0.1$ on $\mathsf{subject}$,
respectively.

A returned result is either infeasible, or a \md\ with threshold pattern on the given $X\rightarrow Y$, for example, one of the result returned by real experiment on \emph{Cora} is:
$$\varphi(\mathsf{author, volume, title\rightarrow venue}, <0.6,0.0,0.8,0.6>)$$
with $\mathsf{support}(\varphi)=0.020$ and  $\mathsf{confidence}(\varphi)=0.562$ both greater than the specified requirements of $\eta_s$ and $\eta_c$ respectively.

\paragraph{Exact Approach Evaluation}

First, we evaluate the performance of pruning by support (\eps)
compared with the original exact algorithm (\ea). As shown in (a)
and (b) in Figure~\ref{exp_eps_time_cite}, \ref{exp_eps_time_cora} and \ref{exp_eps_time_rest},
the \ea, which verifies all the possible candidates, should have the
same cost no matter how $\eta_s$ and $\eta_c$ set. Therefore, the time cost of \ea\ in (a) is exactly the same as that in (b) in all three data sets.

Moreover, the \eps\ achieves significantly lower time cost in all the statistical
distributions, which is only about $1/10$ of that of the \ea. These results demonstrate that our \eps\ approach can prune most of candidates without costly computation.
Note that the time costs of approaches increase linearly with data sizes, which shows the scalability of discovering \mds\ on large data.

To observe more accurately, we also plot the \eps\ time cost in
Figure~\ref{exp_advanced_time_cite}, \ref{exp_advanced_time_cora} and \ref{exp_advanced_time_rest} with the same settings
respectively. According to the pruning strategy, the \eps\
performance is only affected by support requirement $\eta_s$.
In other words, different $\eta_c$ settings take no effect on \eps. Thus, \eps\ has similar time costs in  Figure~\ref{exp_advanced_time_cora} (a) and (b) with the same $\eta_s$ but different $\eta_c$. Similar results can be observed in Figure~\ref{exp_advanced_time_rest} as well.

On the other hand, the \eps\ approach conducts the pruning based on the given requirement of support $\eta_s$. It is natural that a higher $\eta_s$ turns to the better pruning
performance. Therefore, \eps\ with $\eta_s=0.04$ in
Figure~\ref{exp_advanced_time_cite} (a) shows lower time cost, e.g., about
$0.4$s for $150$k, than that of $\eta_s=0.01$ in (b), e.g., $0.6$s
for the same $150$k.
Similar results with different $\eta_s$ are also observed on
\emph{Cora} and \emph{Restaurant}, which are not presented due to the limit of space.

\paragraph{Advanced Approach Evaluation}

Now, we report the performance of advanced pruning and approximation techniques in Figure~\ref{exp_advanced_time_cite}, \ref{exp_advanced_time_cora} and \ref{exp_advanced_time_rest},
including the pruning by both support and confidence (\epsc), the
approximation together with pruning by support (\aps), and the
approximation individually together with pruning by support (\apsi).

First, we study the influence of $\eta_c$ in different approaches.
When the confidence requirement $\eta_c$ is high, e.g., in
Figure~\ref{exp_advanced_time_cora} (b) and \ref{exp_advanced_time_rest} (b), the \epsc\ can remove those low confidence candidates and shows better time performance than other
approaches. On the other hand, when $\eta_c$ is small, e.g.,
$\eta_c=0.15$, we can have larger choices of
$\epsilon\in(0,1-\eta_c)$ such as $\epsilon=0.8$ in
Figure~\ref{exp_advanced_time_cora} (a) and \ref{exp_advanced_time_rest} (a). Thus, the approximation
approaches have lower time cost, especially the \apsi.
According to this analysis, we can choose \epsc\ in practical cases if the requirement $\eta_c$ is high; otherwise, the \apsi\ is preferred in order to achieve lower time costs.

According to the bound condition of approximation
approaches in Theorem~\ref{the_approx}, not only $\epsilon$, but also  the $\eta_s$ affects the
performance. As presented in Figure~\ref{exp_advanced_time_cite} (a), a
higher $\eta_s$ contributes a larger bound condition, which means
the early termination of the program. Thus, approximation approaches
show better performance in Figure~\ref{exp_advanced_time_cite} (a), having $\eta_s = 0.04$,
compared with Figure~\ref{exp_advanced_time_cite} (b), whose $\eta_s = 0.01$.

Note that the bound condition also depends on the distribution
features. A preferred distribution with more tuples in $\bar{\beta}$
can achieve the bound condition and terminate early, such as $50$k
in Figure~\ref{exp_advanced_time_cora} (a) with low time cost.

    \begin{figure*}[t]
    \centering
    \includegraphics[width=0.33\linewidth]{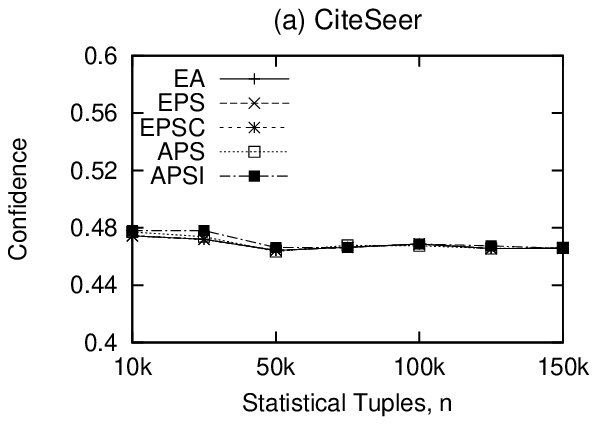}%
    \includegraphics[width=0.33\linewidth]{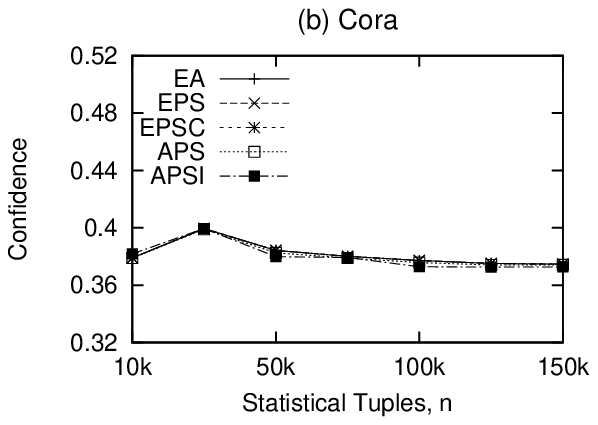}%
    \includegraphics[width=0.33\linewidth]{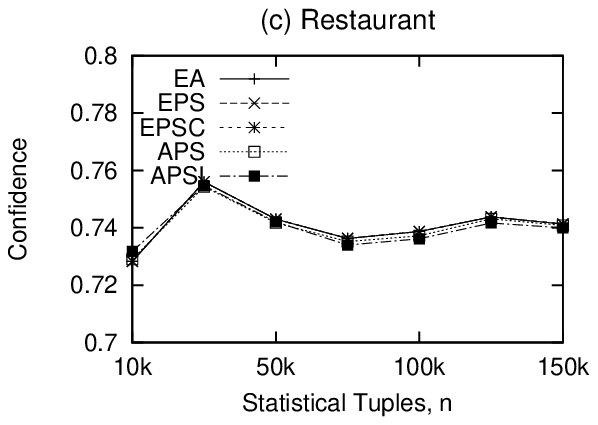}
    \caption{Relative error of approximate confidence}
    \label{exp_md_apprx_confidence}
    \centering
    \includegraphics[width=0.33\linewidth]{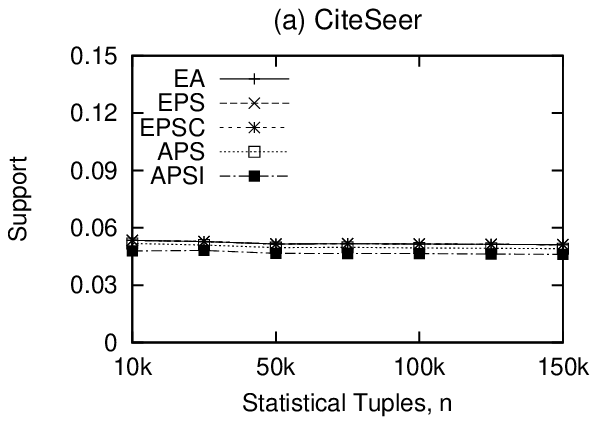}%
    \includegraphics[width=0.33\linewidth]{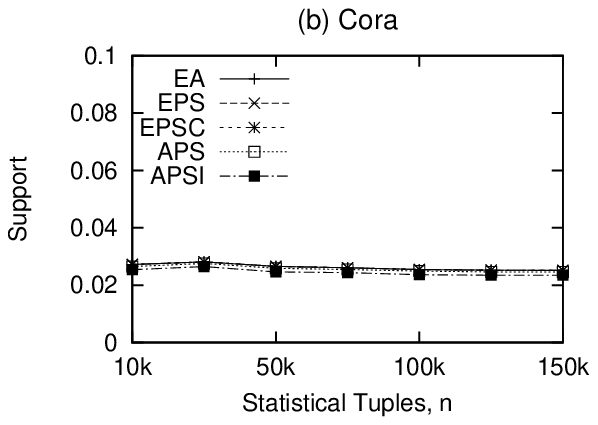}%
    \includegraphics[width=0.33\linewidth]{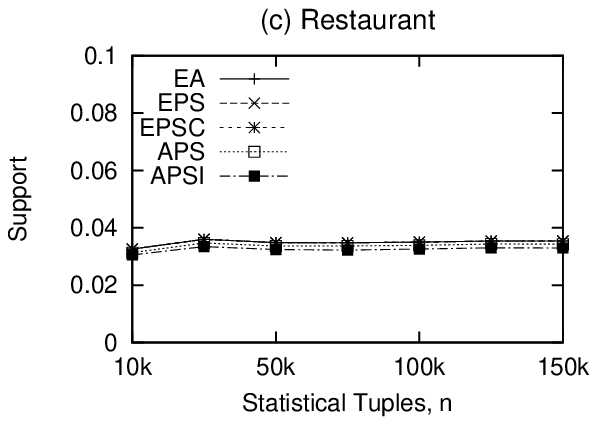}
    \caption{Relative error of approximate support}
    \label{exp_md_apprx_support}
    \end{figure*}

Finally, we evaluate the approximate confidence and support of the
returned \mds\ with $\epsilon=0.8$ on both two datasets in
Figure~\ref{exp_md_apprx_confidence} and \ref{exp_md_apprx_support}.
As we proved in Lemma~\ref{lemma_md_appx_bound}, the error
introduced in approximation approaches is bounded by $\epsilon$ on
both confidence and support. Therefore, in
Figure~\ref{exp_md_apprx_confidence} and \ref{exp_md_apprx_support},
the approximate confidence and support of \aps\ and \apsi\ are very
close to those of exact algorithms.

Consequently, the approximate algorithm can achieve low time cost
(e.g., in Figure~\ref{exp_advanced_time_cora} (a), \ref{exp_advanced_time_rest} (a) and
\ref{exp_advanced_time_cite} (a) with the same setting of
$\epsilon$) without introducing large variation in the confidence
and support measures compared with the exact ones.

\paragraph{Summary}

The experiment results demonstrate that our pruning and approximation techniques can significantly improve the efficiency of discovering \mds. \textbf{i)} The time costs of approaches increase linearly with data sizes, which shows the scalability of discovering \mds\ on large data. \textbf{ii)} The \eps\ approach can significantly reduce the time costs by pruning candidates, compared with the \ea. \textbf{iii)} If
the minimum confidence requirement $\eta_c$ is high, the pruning by confidence
works well. \textbf{iv)} Otherwise, we can employ the approximation approaches
to achieve low time cost.

\section{Conclusions}\label{sect_conclusion}

In this paper, we study the discovery of matching dependencies.
First, we formally define the utility evaluation of matching dependencies by
using support and confidence. Then, we introduce the problem of
discovering the \mds\ with minimum confidence and support
requirements. Both pruning strategies and approximation of the exact
algorithm are studied. The pruning by support can filter out the
candidate patterns with low supports. In addition, if the minimum
confidence requirement is high, the pruning by confidence works
well; otherwise, we can employ the approximation approaches to
achieve low time cost.
The experimental evaluation demonstrates
the performance of proposed methods.

Since this is the first work on discovering the matching
dependencies, there are many aspects of work to develop in the
future. For example, although the current approach can exclude the
attributes that are not necessary to a \md, another issue is to
minimize the number of attributes in the \md. However, the problem
of determining attributes for \fds\ is already
hard~\cite{DBLP:journals/cj/HuhtalaKPT99}, where the matching
similarity thresholds are not necessary to be considered. Moreover, two different \mds\ may cover different dependency semantics, which leads us to the problem of generating \mds\ set. Rather than a single \md, the utility evaluation of a \mds\ set is also interesting. Finally, and most
importantly, more exiting applications of \mds\ are
expected to be explored in the future work.
Finally, along the same line as evaluating \fds~\cite{DBLP:journals/tcs/KivinenM95,DBLP:conf/webdb/NambiarK04}, the \mds\ utility can also be measured by the smallest number of tuples that would have to be removed from the relation in order to eliminate all violations.

\bibliographystyle{abbrv}
{\small%
\bibliography{dependency}

\begin{thebibliography}{10}

\bibitem{DBLP:books/aw/AbiteboulHV95}
S.~Abiteboul, R.~Hull, and V.~Vianu.
\newblock {\em Foundations of Databases}.
\newblock Addison-Wesley, 1995.

\bibitem{DBLP:conf/sigmod/AgrawalIS93}
R.~Agrawal, T.~Imielinski, and A.~N. Swami.
\newblock Mining association rules between sets of items in large databases.
\newblock In {\em SIGMOD Conference}, pages 207--216, 1993.

\bibitem{DBLP:books/sp/dcsa/Batini06}
C.~Batini and M.~Scannapieco.
\newblock {\em Data Quality: Concepts, Methodologies and Techniques}.
\newblock Data-Centric Systems and Applications. Springer, 2006.

\bibitem{DBLP:journals/expert/BilenkoMCRF03}
M.~Bilenko, R.~J. Mooney, W.~W. Cohen, P.~Ravikumar, and S.~E. Fienberg.
\newblock Adaptive name matching in information integration.
\newblock {\em IEEE Intelligent Systems}, 18(5):16--23, 2003.

\bibitem{DBLP:conf/icde/BohannonFGJK07}
P.~Bohannon, W.~Fan, F.~Geerts, X.~Jia, and A.~Kementsietsidis.
\newblock Conditional functional dependencies for data cleaning.
\newblock In {\em ICDE}, pages 746--755, 2007.

\bibitem{DBLP:conf/icde/BravoFGM08}
L.~Bravo, W.~Fan, F.~Geerts, and S.~Ma.
\newblock Increasing the expressivity of conditional functional dependencies
  without extra complexity.
\newblock In {\em ICDE}, pages 516--525, 2008.

\bibitem{DBLP:conf/vldb/BravoFM07}
L.~Bravo, W.~Fan, and S.~Ma.
\newblock Extending dependencies with conditions.
\newblock In {\em VLDB}, pages 243--254, 2007.

\bibitem{DBLP:journals/tods/CaldersNW02}
T.~Calders, R.~T. Ng, and J.~Wijsen.
\newblock Searching for dependencies at multiple abstraction levels.
\newblock {\em ACM Trans. Database Syst.}, 27(3):229--260, 2002.

\bibitem{DBLP:journals/pvldb/ChiangM08}
F.~Chiang and R.~J. Miller.
\newblock Discovering data quality rules.
\newblock {\em PVLDB}, 1(1):1166--1177, 2008.

\bibitem{DBLP:conf/sigmod/Cohen98}
W.~W. Cohen.
\newblock Integration of heterogeneous databases without common domains using
  queries based on textual similarity.
\newblock In {\em SIGMOD Conference}, pages 201--212, 1998.

\bibitem{DBLP:conf/vldb/CongFGJM07}
G.~Cong, W.~Fan, F.~Geerts, X.~Jia, and S.~Ma.
\newblock Improving data quality: Consistency and accuracy.
\newblock In {\em VLDB}, pages 315--326, 2007.

\bibitem{DBLP:journals/tkde/ElmagarmidIV07}
A.~K. Elmagarmid, P.~G. Ipeirotis, and V.~S. Verykios.
\newblock Duplicate record detection: A survey.
\newblock {\em IEEE Trans. Knowl. Data Eng.}, 19(1):1--16, 2007.

\bibitem{DBLP:conf/pods/Fan08}
W.~Fan.
\newblock Dependencies revisited for improving data quality.
\newblock In {\em PODS}, pages 159--170, 2008.

\bibitem{DBLP:conf/icde/FanGLX09}
W.~Fan, F.~Geerts, L.~V.~S. Lakshmanan, and M.~Xiong.
\newblock Discovering conditional functional dependencies.
\newblock In {\em ICDE}, pages 1231--1234, 2009.

\bibitem{DBLP:journals/pvldb/FanLJM09}
W.~Fan, J.~Li, X.~Jia, and S.~Ma.
\newblock Reasoning about record matching rules.
\newblock {\em PVLDB}, 2009.

\bibitem{DBLP:journals/pvldb/FanMHLW08}
W.~Fan, S.~Ma, Y.~Hu, J.~Liu, and Y.~Wu.
\newblock Propagating functional dependencies with conditions.
\newblock {\em PVLDB}, 1(1):391--407, 2008.

\bibitem{DBLP:journals/pvldb/GolabKKSY08}
L.~Golab, H.~J. Karloff, F.~Korn, D.~Srivastava, and B.~Yu.
\newblock On generating near-optimal tableaux for conditional functional
  dependencies.
\newblock {\em PVLDB}, 1(1):376--390, 2008.

\bibitem{DBLP:conf/www/GravanoIKS03}
L.~Gravano, P.~G. Ipeirotis, N.~Koudas, and D.~Srivastava.
\newblock Text joins in an rdbms for web data integration.
\newblock In {\em WWW}, pages 90--101, 2003.

\bibitem{DBLP:journals/cj/HuhtalaKPT99}
Y.~Huhtala, J.~K{\"a}rkk{\"a}inen, P.~Porkka, and H.~Toivonen.
\newblock Tane: An efficient algorithm for discovering functional and
  approximate dependencies.
\newblock {\em Comput. J.}, 42(2):100--111, 1999.

\bibitem{DBLP:conf/sigmod/IlyasMHBA04}
I.~F. Ilyas, V.~Markl, P.~J. Haas, P.~Brown, and A.~Aboulnaga.
\newblock Cords: Automatic discovery of correlations and soft functional
  dependencies.
\newblock In {\em SIGMOD Conference}, pages 647--658, 2004.

\bibitem{DBLP:journals/jamds/KingL03}
R.~S. King and J.~J. Legendre.
\newblock Discovery of functional and approximate functional dependencies in
  relational databases.
\newblock {\em JAMDS}, 7(1):49--59, 2003.

\bibitem{DBLP:journals/tcs/KivinenM95}
J.~Kivinen and H.~Mannila.
\newblock Approximate inference of functional dependencies from relations.
\newblock {\em Theor. Comput. Sci.}, 149(1):129--149, 1995.

\bibitem{DBLP:conf/icde/KoudasSSV09}
N.~Koudas, A.~Saha, D.~Srivastava, and S.~Venkatasubramanian.
\newblock Metric functional dependencies.
\newblock In {\em ICDE}, pages 1275--1278, 2009.

\bibitem{DBLP:conf/kdd/McCallumNU00}
A.~McCallum, K.~Nigam, and L.~H. Ungar.
\newblock Efficient clustering of high-dimensional data sets with application
  to reference matching.
\newblock In {\em KDD}, pages 169--178, 2000.

\bibitem{DBLP:conf/webdb/NambiarK04}
U.~Nambiar and S.~Kambhampati.
\newblock Mining approximate functional dependencies and concept similarities
  to answer imprecise queries.
\newblock In {\em WebDB}, pages 73--78, 2004.

\bibitem{DBLP:journals/csur/Navarro01}
G.~Navarro.
\newblock A guided tour to approximate string matching.
\newblock {\em ACM Comput. Surv.}, 33(1):31--88, 2001.

\bibitem{DBLP:journals/ida/Scheffer05}
T.~Scheffer.
\newblock Finding association rules that trade support optimally against
  confidence.
\newblock {\em Intell. Data Anal.}, 9(4):381--395, 2005.

\end{thebibliography}
}%



%
%
%
%
%
%
%
%
%
%


\end{document}